\documentclass[preprint]{imsart}

\RequirePackage[OT1]{fontenc}
\RequirePackage{amsthm,amsmath}
\RequirePackage[numbers]{natbib}
\RequirePackage[colorlinks,citecolor=blue,urlcolor=blue]{hyperref}

\usepackage{graphicx}
\usepackage{amssymb}
\usepackage{epstopdf}
\usepackage{bbm}
\usepackage{amsmath,amssymb}
\usepackage{subfigure}
\usepackage{thmtools,thm-restate}
\usepackage{natbib}

\startlocaldefs
\numberwithin{equation}{section}
\theoremstyle{plain}
\newtheorem{thm}{Theorem}[section]
\newtheorem{prop}[thm]{Proposition}
\newtheorem{defn}{Definition}[section]
\newtheorem{remark}{Remark}[section]
\newtheorem{lem}[thm]{Lemma}

\def\bigE{{\mathbb{E}}}
\DeclareMathOperator{\sgn}{sgn}
\endlocaldefs

\begin{document}

\begin{frontmatter}
\title{When Multiplicative Noise Stymies Control}
\runtitle{When Multiplicative Noise Stymies Control}

\begin{aug}
\author{\fnms{Jian} \snm{Ding}\ead[label=e1]{jianding@galton.uchicago.edu}}\thanksref{t1},
\author{\fnms{Yuval} \snm{Peres}\ead[label=e2]{peres@microsoft.com}}\thanksref{t2}, 
\author{\fnms{Gireeja} \snm{Ranade}\ead[label=e3]{giranade@microsoft.com}}\thanksref{t2}
\and
\author{\fnms{Alex} \snm{Zhai}\ead[label=e4]{azhai@stanford.edu}}\thanksref{t3}

\thankstext{t1}{University of Chicago, Department of Statistics, \printead{e1}}
\thankstext{t2}{Microsoft Research, Redmond, \printead{e2}, \printead*{e3}}
\thankstext{t3}{Stanford University, Department of Mathematics,\printead{e4}}

\runauthor{J. Ding et al.}



\end{aug}

%
%
%
%
%
%

\begin{abstract}
We consider the stabilization of an unstable discrete-time linear system that is observed over a channel corrupted by continuous multiplicative noise. Our main result shows that if the system growth is large enough, then the system cannot be stabilized in a second-moment sense. This is done by showing that the probability that the state magnitude remains bounded must go to zero with time. Our proof technique recursively bounds the conditional density of the system state (instead of focusing on the second moment) to bound the progress the controller can make. This sidesteps the difficulty encountered in using the standard data-rate theorem style approach; that approach does not work because the mutual information per round between the system state and the observation is potentially unbounded.

It was known that a system with multiplicative observation noise can be stabilized using a simple memoryless linear strategy if the system growth is suitably bounded. In this paper, we show that while memory cannot improve the performance of a linear scheme, a simple non-linear scheme that uses one-step memory can do better than the best linear scheme.
\end{abstract}



\end{frontmatter}

\section{Introduction}

We consider the control and stabilization of a system observed over a multiplicative noise channel. Specifically, we analyze the following system, $\mathcal{S}_{a}$, with initial state $X_{0} \sim \mathcal{N}(0,1)$:
\begin{align}
\begin{split}
X_{n+1} &= a \cdot X_{n} - U_{n}, \\
Y_{n} &= Z_{n} \cdot X_{n}. \label{eq:systema}
\end{split}
\end{align}
In the preceding formulation, the system state is represented by $X_{n}$ at time $n$, and the control $U_{n}$ can be any function of the current and previous observations $Y_{0}$ to $Y_{n}$. The $Z_{n}$'s are i.i.d.\ random variables with a known continuous distribution. The realization of the noise $Z_{n}$ is unknown to the controller, much like the fading coefficient (gain) of a channel might be unknown to the transmitter or receiver in non-coherent communication. The constant $a$ captures the growth of the system. The controller's objective is to stabilize the system in the second-moment sense, i.e. to ensure that $\sup_{n}\bigE[|X_{n}|^{2}] < \infty$. Our objective is to understand the largest growth factor $a$ that can be tolerated for a given distribution on $Z_{n}$. Fig.~\ref{fig:system} represents a block diagram for this system. 

\begin{figure}[pthb]
\centering
    \includegraphics[width=.4\textwidth]{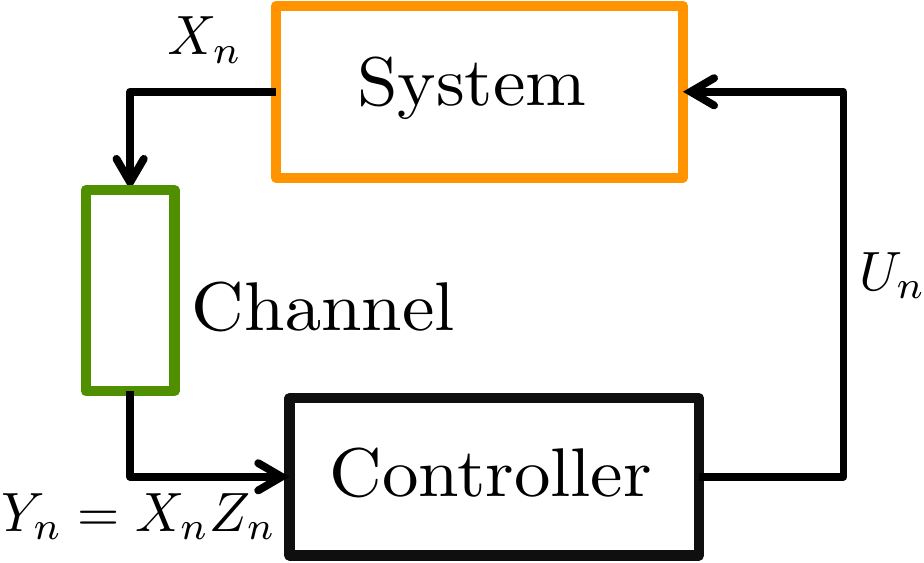} 
   \caption{The state $X_{n}$ is observed over a multiplicative noise channel $Y_{n} = X_{n}Z_{n}$.}
   \label{fig:system}
\end{figure}

Our main theorem provides an impossibility result for stabilizing the system $\mathcal{S}_{a}$. 

{\thm{Let the $Z_{n}$ be i.i.d.~random variables with finite mean and variance and with bounded density $f_{Z}(z)= e^{-\phi(z)}$, where $\phi(\cdot)$ is a polynomial of even degree with positive leading coefficient. Then, there exists $a \in \mathbb{R}$, $a < \infty$ such that $|X_{n}|$ in~\eqref{eq:systema} satisfies $\mathbb{P}(|X_{n}|< M) \rightarrow 0$ for all $M < \infty$.\label{thm:intro1}}}

Thm.~\ref{thm:main} generalizes this result to a larger class of distributions for $Z_{n}$. Note that the conditions on $\phi(\cdot)$ in Thm.~\ref{thm:intro1} are satisfied by \mbox{$Z_{n}\sim \mathcal{N}(1, \sigma^{2})$.} 


We also discuss a few sufficient conditions for second-moment stability of the system in this paper. When $Z_{n}$ has mean $1$ and variance $\sigma^{2}$, we observe that that a system growth of  $a^* = \sqrt{1 + \frac{1}{\sigma^2}}$ can be stabilized in the second-moment sense using a simple \emph{linear} strategy (Prop.~\ref{thm:secmomach}). Further, we show that the best linear strategy to control the system $\mathcal{S}_{a}$ in~\eqref{eq:systema} is memoryless (Thm.~\ref{thm:greedy}). Our second main result (Thm.~\ref{thm:nonlinimprovement}) shows that a \emph{non-linear} controller can improve on the performance of the best linear strategy. We state this here for the case where $Z_{n} \sim \mathcal{N}(1,1)$.

{\thm{Let $Z_n \sim \mathcal{N}(1,1)$. Then the system $\mathcal{S}_{a}$ in~\eqref{eq:systema} with $a \leq \sqrt{2}$ can be stabilized in the second-moment sense by a linear control strategy. Further, there exists $a > \sqrt{2}$ for which a non-linear controller can stabilize the system in a second-moment sense.}}

In particular, there exists a non-linear strategy with memory that can stabilize $\mathcal{S}_{a}$ in a second-moment sense with $a = \sqrt{2} + 1.6 \times 10^{-3}$. 

We further believe that non-linear schemes without memory cannot stabilize the system for $a < a^{*}$, and some evidence in this direction is provided in Thm.~\ref{thm:nocontraction}. Finally, in the case where the $Z_{n}$ have mean zero, a linear strategy cannot stabilize the system in the second-moment sense for any growth factor $a$ (Thm.~\ref{thm:zeromean}), but a non-linear scheme with memory can stabilize it for some value of the growth factor $a$ (Thm.~\ref{thm:nonlinimprovement2}).


\subsection{Model motivation}
Multiplicative noise on the observation channel can model the effects of a fast-fading communication channel (rapidly changing channel gain), as well as the impact of sampling and quantization errors~\cite{meyrdigital, gireejaAllerton}. A more detailed discussion of multiplicative noise models is available in~\cite{gireejabeast}. 

We illustrate below how synchronization or sampling errors can lead to multiplicative noise, following a discussion from~\cite{gireejaAllerton}. Consider the nearly trivial continuous-time system, %
$$\dot{X}(t) = a \cdot X(t), $$
which is sampled at regular intervals of $t_{0}$. The difference equation corresponding to the state at the $n$th time step is given by
$X_{n+1} = e^{a t_{0}} \cdot X_{n}.$
However, in the presence of synchronization error the $n$th sample, $Y_{n}$, might be collected at time $nt_{0} + \Delta $ instead of precisely at $nt_{0}$. Then,
\begin{align}
Y_{n} = e^{a(nt_{0}+\Delta )}X(0) = Z_{n} X_{n}, \nonumber
\end{align}
where $Z_{n}$ is a continuous random variable, since the jitter $\Delta$ is a continuous random variable.

\subsection{Proof approach}
We introduce a new converse approach in the proof of Thm.~\ref{thm:main}; instead of focusing on the second-moment, our proof bounds the density of the state and thus shows the instability of any moment of the state. We believe these techniques are a primary contribution of the work.

A key element of the proof is that a ``genie'' observes the state of the system and provides a quantized version of the logarithm of the state to the controller at each time as extra side-information in addition to the multiplicative noise observation. This side-information bounds the state in intervals of size $2^{-k}$ (with $k$ increasing as time increases). We know from results on non-coherent communication~\cite{lapidothLogLog} and carry-free models~\cite{carryfree} that only the order of magnitude of the message can be recovered from a transmission with multiplicative noise. As a result, this side-information does not effectively provide much extra information, but it allows us to quantify the rate at which the controller may make progress.

\subsection{Related work}
Our problem is connected to the body of work on data-rate theorems and control with communication constraints as studied in~\cite{wong1997systems, tatikonda, nair2007feedback,minero2009data}. These data-rate theorems tell us that a noiseless observation data rate $R> \log|a|$ is necessary and sufficient to stabilize a system in the second-moment sense. Our setup considers multiplicative noise on the observation channel instead of observations over a noiseless but rate-limited channel. Paralleling the data-rate theorems, Prop.~\ref{thm:secmomach} provides a control strategy that can stabilize the system when $\frac{1}{2} \log (1 +\frac{1}{\sigma^{2}}) > \log |a|$ for $Z_{n}$ with mean $1$ and variance $\sigma^{2}$.

Our problem is also inspired by the intermittent Kalman filtering problem~\cite{sinopoli2004kalman, parkIntermittent}, as well the problem of control over lossy networks~\cite{schenato2007foundations,eliaFading} (i.e. estimation and control over Bernoulli multiplicative noise channels). The setup in our paper generalizes those setups to consider a general continuous multiplicative noise on the observation.

The uncertainty threshold principle~\cite{uncertaintyThreshold} considers a systems with Gaussian uncertainty on the system growth factor and the control gain, and provides limits for when the system is stabilizable in a second-moment sense. Our work complements this result by considering uncertainty on the observation gain.

A related problem is that of estimating a linear system over multiplicative noise. While early work on this had been limited to exploring linear estimation strategies~\cite{rajasekaran1971optimum, tugnait1981stability}, some recent work show a general converse result for the estimation problem over multiplicative noise for both linear and non-linear strategies~\cite{gireejaAllerton}. We note that our problem can also be interpreted as an ``active'' estimation problem for $X_{0}$, and our impossibility result applies to both linear and non-linear control strategies. However, techniques from the estimation converse result or the data-rate theorems do not work for our setup here. Unlike the estimation problem, we cannot describe the distribution of $X_{n}$ in our problem since the control $U_{n}$ is arbitrary. For the same reason, we also cannot bound the range of $X_{n}$ or the rate across the observation channel to use a data-rate theorem approach.

Some of our results and methods are summarized in~\cite{tiger}.


\section{Problem statement}

Consider the system $\mathcal{S}_{a}$ in~\eqref{eq:systema}. For simplicity, let the initial state $X_{0}$ be distributed as $X_{0} \sim \mathcal{N}(0,1)$. Let $Z_{n}$ be i.i.d. random variables with finite second moment and bounded density $f_{Z}(z) = e^{-\phi(z)}$. Without loss of generality, we will use the scaling $\bigE Z_{n} = 1$ and $\mathrm{Var}(Z_{n}) = \sigma^{2}$. The notation $Z_n$, $f_Z$, $\phi$, and $\sigma$ defined here will be used throughout the paper.

We introduce two definitions for stability of the system. The first is the notion of stability that is most commonly studied in control theory, i.e. second-moment stability.

{\defn{The system $\mathcal{S}_{a}$ in~\eqref{eq:systema} is said to be \textbf{second-moment} stabilizable if there exists an adapted control strategy $U_{0}, \cdots, U_{n}$ (a control strategy where $U_k$ is a function of $Y_0, \cdots , Y_k$ for each $0 \le k \le n$) such that
    \[ \sup_{n} \bigE[|X_{n}|^{2}] < \infty. \]
}\label{def:sms}}
{\defn{We say the controller can keep the system $\mathcal{S}_{a}$ in~\eqref{eq:systema} \textbf{tight} if for every $\epsilon$ and for every $n$ there exists an adapted control strategy $U_{0}, \cdots, U_{n}$, and there exist $M_{\epsilon}, N_{\epsilon}<\infty$ such that
\[ \mathbb{P}(|X_{n}| < M_{\epsilon}) \geq 1 - \epsilon, \]
 for $n > N_{\epsilon}$.} \label{def:tight}}


\section{Linear schemes} \label{sec:linear-schemes}

This section first provides a simple memoryless linear strategy that can stabilize the system in a second-moment sense in Prop.~\ref{thm:secmomach}. We show in Thm.~\ref{thm:greedy} that this strategy is optimal among linear strategies. In Thm.~\ref{thm:zeromean}, we highlight the limitations of linear strategies by showing that when $\bigE Z_n = 0$, linear strategies cannot stabilize the system for any growth factor $a > 1$. Finally, we consider stability in the sense of keeping the system tight and provide a scheme that achieves this in Thm.~\ref{thm:tightach}.

{\prop[A linear memoryless strategy]{The controller given by $U_{n} = d^{*} Y_{n}$ where $d^{*} = \frac{a}{1 + \sigma^{2}}$, can stabilize the system $\mathcal{S}_{a}$ in~\eqref{eq:systema} in a second-moment sense (Def.~\ref{def:sms}) if $a \le a^{*}$, where $a^{*} =  \sqrt{1 + \frac{1}{\sigma^{2}}}$. \label{thm:secmomach}} }

\begin{proof}
The above strategy gives us $X_{n+1} = (a - d^*Z_n)X_n$. Since $Z_n$ is independent of $X_n$, we can write:
\begin{align*}
  \bigE[X_{n+1}^{2}] &= \bigE[(a - d^*Z_{n})^{2}] \bigE[X_{n}^{2}] = (a^2 - 2ad^* + (d^*)^{2}(1 + \sigma^{2})) \bigE[X_{n}^{2}] \\
  &= \frac{a^2 \sigma^2}{1 + \sigma^2} \cdot \bigE[X_{n}^{2}].
\end{align*}

Under this control strategy $\sup_{n} \bigE[|X_{n}|^{2}]$ is bounded if and only if \\ \mbox{$a^{2}\leq1 + \frac{1}{\sigma^{2}}$}. 
\end{proof}

Note that the above controller is linear in that $U_n$ is a linear function of the $Y_i$ and memoryless in that $U_n$ depends only on $Y_n$ and not $Y_i$ for $i < n$. We might expect an improvement in the achievable performance of a linear strategy if we also allow it to use memory, i.e. the past $Y_{n}$'s. However, it turns out that the optimal linear strategy is in fact memoryless.

\begin{thm}
The control strategy given by $U_{n} = d^{*} Y_{n}$ where $d^{*} = \frac{a}{1 + \sigma^{2}}$ is the optimal linear strategy to stabilize $\mathcal{S}_a$ in a second-moment sense, in particular, for all $a >  \sqrt{1 + \tfrac{1}{\sigma^{2}}}$ the system $\mathcal{S}_{a}$ in~\eqref{eq:systema} cannot be second-moment stabilized (Def.~\ref{def:sms}) using a linear strategy.
\label{thm:greedy}
\end{thm}

\begin{proof}
Suppose the system $\mathcal{S}_a$ evolves following some linear strategy of the form $U_n = \sum_{i = 1}^n \alpha_{n, i} Y_i$.

We define a system $\widetilde{\mathcal{S}}$ such that $\widetilde{X}_{n}$ that evolves in parallel with $X_{n}$ and tracks the behavior of the strategy  $U_{n} = d^{*} Y_{n}$. Formally, $\widetilde{\mathcal{S}}$ is defined as:
\begin{align*}
\widetilde{X}_{0} &= X_{0} \\
\widetilde{Y}_{n} &= Z_{n}\widetilde{X}_{n} \\
\widetilde{X}_{n+1} &= a \widetilde{X}_{n} -  d^{*} \widetilde{Y}_{n},
\end{align*}
where the $Z_{n}$'s are the same as those acting on $X_{n}$. Then, we can write
\begin{align}
\widetilde{X}_{n+1}  &= \widetilde{X}_{n} (a - d^{*}Z_{n})\nonumber\\
&= \widetilde{X}_{n-m} \prod_{i=n-m}^{n}(a - d^{*}Z_{i}). \label{eq:tildeproduct}
\end{align}

We will show that $\bigE[|\widetilde{X}_{n}|^{2}]$ is the minimum achievable second moment at any time $n$. Since $\bigE[|\widetilde{X}_{n}|^{2}] < \infty$ only when $a \leq  \sqrt{1 + \frac{1}{\sigma^{2}}}$, we are done once we show this.

Our approach is to inductively show that $\bigE[(X_{n} - \widetilde{X}_{n}) \widetilde{X}_{n}] = 0$ for all $n$ and for any linear control strategy applied to the system $\mathcal{S}$, from which it follows that
\[ \bigE[X_n^2] = \bigE[\widetilde{X}_n^2] + \bigE[(X_{n} - \widetilde{X}_n)^2] \ge \bigE[\widetilde{X}_n^2]. \]

Base case: $n=0$ is trivially true, since $X_{0} = \widetilde{X}_{0}$.
Assume that our hypothesis is true for $n=k$. Now, consider $n = k+1$:
\begin{align}
&\bigE[(X_{k+1}-\widetilde{X}_{k+1})\widetilde{X}_{k+1}] \nonumber \\
=~& a \bigE[(X_{k}-\widetilde{X}_{k})\widetilde{X}_{k+1}]  - \bigE\left[\left(\sum_{i=0}^{k}\alpha_{k,i}Y_{i} \right) \widetilde{X}_{k+1}\right] + \bigE[d^*\widetilde{Y}_{k}\widetilde{X}_{k+1}]. \label{eq:expectsplit}
\end{align}
We will show that all three expectations in the final expression are zero. The first term in~\eqref{eq:expectsplit} is
\begin{align*}
\bigE[(X_{k}-\widetilde{X}_{k})\widetilde{X}_{k+1}] &= \bigE[(X_{k}-\widetilde{X}_{k})(a \widetilde{X}_{k} -  d^{*} \widetilde{Y}_{k})] \\
&= - d^{*} \bigE[(X_{k}-\widetilde{X}_{k})\widetilde{Y}_{k}],
\end{align*}
by the induction hypothesis. Because $Z_k$ is independent of $X_k$ and $\widetilde{X}_k$, we may compute the above expectation as
\begin{align}
\bigE[(X_{k}-\widetilde{X}_{k})\widetilde{Y}_{k}] &= \bigE[(X_{k}-\widetilde{X}_{k})Z_{k}\widetilde{X}_{k}] \nonumber \\
&=\bigE[Z_{k}]\bigE[(X_{k}-\widetilde{X}_{k})\widetilde{X}_{k}] \nonumber \\
&=0. \label{eq:prod1}
\end{align}
To handle the second term, for each $1 \le i \le k$ we can apply~\eqref{eq:tildeproduct} to obtain
\begin{align}
\bigE[Y_i\widetilde{X}_{k+1}] =&~\bigE\left[Z_iX_i\widetilde{X}_i \prod_{j=i}^{k+1}(a - d^{*}Z_{j})\right] \nonumber\\
=&~\bigE[Z_i (a-d^{*}Z_i)]\bigE\left[X_i\widetilde{X}_i\prod_{j=i+1}^{k+1}(a - d^{*}Z_{j})\right]\nonumber \\
=&~0, \label{eq:prod2}
\end{align}
where again we have used the independence of $Z_i$ from the other
terms in the product, and $\bigE[Z_i (a - d^*Z_i)] = 0$ from the
definition of $d^*$. Finally, the last term may be computed in a
similar manner as
\begin{align}
\bigE[\widetilde{Y}_{k}\widetilde{X}_{k+1}] &= \bigE[Z_k\widetilde{X}_{k}(a \widetilde{X}_{k} - d^{*}Z_k\widetilde{X}_{k})]\nonumber \\
&=\bigE[Z_{k}(a  - d^{*}Z_{k})]\bigE[\widetilde{X}_{k}^{2}] \nonumber \\
&= 0. \label{eq:prod3}
\end{align}
by the definition of $d^{*}$.

Equations \eqref{eq:prod1}, \eqref{eq:prod2}, and \eqref{eq:prod3}, establish that all three terms in \eqref{eq:expectsplit} are zero. Hence, $\bigE[(X_{n} - \widetilde{X}_{n}) \widetilde{X}_{n}] = 0$ for all $n$, and we are done. \qedhere

\end{proof}

A similar analysis illustrates the limitations of linear strategies when $\bigE Z_n = 0$, in contrast with nonlinear strategies to be
described in the next section.

\begin{thm} \label{thm:zeromean}
  Suppose that instead of $\bigE Z_n = 1$, we have $\bigE Z_n =
  0$. Then, for all $a > 1$, the system $\mathcal{S}_{a}$ in~\eqref{eq:systema} cannot be second-moment
  stabilized using a linear strategy. In other words, linear
  strategies cannot tolerate any growth in the system.
\end{thm}

\begin{proof}
  Suppose the system $\mathcal{S}_a$ evolves following some linear
  strategy of the form $U_n = \sum_{i = 1}^n \alpha_{n, i} Y_i$.

  We will show by induction that for each $n$, we may write $X_n =
  W_nX_0$, where $W_n$ is a function of $Z_0, Z_1, \ldots , Z_{n -
    1}$, and $\bigE W_n = a^n$. Clearly, this holds for $n = 0$ with
  $W_0 = 1$. For the inductive step, note that
  \begin{align*}
    X_{n + 1} &= a X_n - U_n = \left(aW_n - \sum_{i = 1}^n \alpha_{n,
      i} Z_iW_i \right) X_0,
  \end{align*}
  so we may take $W_{n + 1} = aW_n - \sum_{i = 1}^n \alpha_{n, i}
  Z_iW_i$. Since $Z_i$ is independent of $W_i$ for each $i$, we have
  \[ \bigE W_{n + 1} = a \bigE W_n - \sum_{i = 1}^n \alpha_{n, i} (\bigE Z_i) (\bigE W_i) = a^{n + 1}, \]
  completing the induction. It follows that
  \[ \bigE \left[X_n^2\right] = \bigE \left[W_n^2\right] \cdot \bigE \left[X_0^2\right] \ge (\bigE W_n)^2 \cdot \bigE \left[X_0^2\right] = a^{2n} \cdot \bigE \left[X_0^2\right], \]
  and so $\bigE \left[X_n^2\right]$ grows without bound when $a > 1$.
\end{proof}

Finally, the next theorem considers the weaker sense of stability of keeping the system tight, which is the sense of stability that the impossibility results in Section~\ref{sec:converse} use.

\begin{thm}
Suppose that the density function $f_Z$ of $Z_n$ is bounded, and consider linear memoryless strategies of the form $U_n = ad \cdot Y_n$ for a constant $d > 0$. Let $d^{\star} =  \mathrm{argmin}_{d}~\bigE[\log|1 - d\cdot Z_{n}|]$ and $a^{\star} = e^{- \bigE[\log|1 - d^{\star}\cdot Z_{n}|]}$. If $d = d^\star$, then the system $\mathcal{S}_{a}$ in~\eqref{eq:systema} can be kept tight (Def.~\ref{def:tight}) provided that $|a| < a^\star$. Further, no such strategy can keep the system tight if $|a| \geq a^\star$. \label{thm:tightach}
\end{thm}

\begin{proof}
Applying the control law $U_{n} = a d Y_{n}$, we calculate that
\begin{align*}
X_{ n} &= a (1 - dZ_{n-1}) X_{ n-1}\\
&=a^{n} \prod_{i=0}^{n-1}(1 - dZ_{i}) X_{0}.
\end{align*}
Let $W_i = \log |1 - dZ_i|$, and let $S_n = \sum_{i = 1}^n (W_i + \log
|a|)$. Taking logarithms gives us
\begin{align*}
\log |X_{ n}| &= S_n + \log |X_{0}|.
\end{align*}
Note that
\[ \lim_{n \rightarrow \infty} \frac{1}{\sqrt{n}} \log |X_{n}| = \lim_{n \rightarrow \infty} \frac{1}{\sqrt{n}} S_n \]
almost surely, so as will be seen shortly, it suffices to analyze $S_n$.

Take $C$ to be an upper bound on the density of $Z_i$. Then, we have
\[ \mathbb{P}(W_i < -t) = \mathbb{P}(|1 - dZ_i| < e^{-t}) \le \frac{C \cdot e^{-t}}{d}, \]
so that $W_i$ has an exponentially decaying left tail. Similarly,
\[ \mathbb{P}(W_i > t) = \mathbb{P}(|1 - dZ_i| > e^t) \le \frac{\bigE (1 - dZ_i)^2}{e^{2t}}, \]
so $W_i$ also has an exponentially decaying right tail.

Thus, $W_i$ has finite first and second moments. Let $\mu_d$ and
$\sigma_d$ denote the mean and variance of $W_i$,
respectively. Defining $\widetilde{S}_n = \sum_{i = 1}^n (W_i -
\mu_d)$, the central limit theorem gives us that
\begin{equation} \label{eq:clt}
  \frac{1}{\sqrt{n}} \widetilde{S}_n \stackrel{\mathcal{D}}{\longrightarrow} \mathcal{N}\left(0, \sigma_d^2\right) \text{ as } n \rightarrow \infty.
\end{equation}

If $|a| < a^\star$ and we take $d = d^\star$, then we see that $\log |a| < \log a^\star = -\mu_{d^\star} = -\mu_d$. Thus, there exists $\epsilon > 0$ such that $\log|a| + \mu_{d} < -2\epsilon.$ Using the union bound we then have:
\begin{align*}
\mathbb{P}\left( \log |X_{n}| \geq -n \epsilon \right) \leq \mathbb{P} \left( \log |X_{0}| > n\epsilon\right) +\mathbb{P} \left( S_{n} \geq  -2n\epsilon\right). 
\end{align*}
We have that $\mathbb{P} \left( \log |X_{0}| > n\epsilon\right) \to 0$ as $n\to \infty$, and also by the law of large numbers $\mathbb{P} \left( S_{n} \geq  -2n\epsilon\right) \to 0$ almost surely. Hence, $\mathbb{P}\left( \log |X_{n}| < -n \epsilon \right) \to 1$ and the system is kept tight.


On the other hand, suppose that $|a| \ge a^\star$. Then, we have $\log |a| \ge \log a^\star = -\mu_{d^\star} \ge -\mu_d$, so $S_n \ge \widetilde{S}_n$. Consider $\delta > 0$. For $n$ large enough we have that:
\begin{align*}
&\mathbb{P}\left( \log |X_{n}| \leq n^{\frac{1}{4}} \right) \leq \mathbb{P}\left( \log |X_{n} \leq \delta \sqrt{n} \right)\\
& \leq \mathbb{P}\left( \log |X_{0}| \leq -\delta \sqrt{n} \right) +  \mathbb{P}\left( S_{n}  \leq 2\delta \sqrt{n} \right)\\
& \leq \mathbb{P}\left( \log |X_{0}| \leq -\delta \sqrt{n} \right) +  \mathbb{P}\left( \widetilde{S}_{n}  \leq 2\delta \sqrt{n} \right),
\end{align*}
where we used the union bound and the fact that $S_n \ge \widetilde{S}_n$ to get the two inequalities.
Now, $\mathbb{P}\left( \log |X_{0}| \leq -\delta \sqrt{n} \right) \to 0$ as $n \to \infty$ and $\mathbb{P}\left( \widetilde{S}_{n}  \leq 2\delta \sqrt{n} \right) \to \Phi\left( \frac{2\delta}{\sigma_{d}}\right)$, by~\eqref{eq:clt}.
Hence, 
$$\limsup_{n\to\infty} \mathbb{P}\left( \log |X_{n}| \leq n^{\frac{1}{4}} \right) \leq \Phi\left( \frac{2\delta}{\sigma_{d}}\right),$$
which gives that
$$\limsup_{n\to\infty} \mathbb{P}\left( \log |X_{n}| > n^{\frac{1}{4}} \right) \geq \frac{1}{2}.$$
Thus, in this case the system is not kept tight.

\end{proof}

\section{Non-linear schemes} \label{sec:nonlinear-schemes}

In the previous section, we focused on linear strategies, where $U_n$
is taken to be a linear combination of $Y_i$ for $0 \le i \le n$. We
now consider whether more general strategies can do
better. Thm.~\ref{thm:nonlinimprovement} shows that when $Z_n$ is
Gaussian, a perturbation of the linear strategy indeed does better in the second-moment sense. 
(The same result should hold for rather general
  $Z_n$; see Remark~\ref{rem:nonlinimprovement}.)
In the setting where $\bigE Z_n = 0$, Thm.~\ref{thm:nonlinimprovement2}
exhibits a nonlinear strategy that achieves a non-trivial growth factor $a >
1$. This contrasts with Thm.~\ref{thm:zeromean}, which showed that
linear strategies cannot achieve any gain in this setting. In both
Thm.~\ref{thm:nonlinimprovement} and
Thm.~\ref{thm:nonlinimprovement2}, improvement is achieved by taking
into account information from the previous round while choosing the
control.

On the other hand, Thm.~\ref{thm:nocontraction} shows that when $a >
a^* = \sqrt{1 + \frac{1}{\sigma^2}}$, for any memoryless strategy (in
the sense that $U_n$ is a function of only $Y_n$), we cannot guarantee
for all distributions of $X_n$ that $\bigE \left[X_{n + 1}^2\right] \le \bigE
\left[X_n^2\right]$. This suggests that in the memoryless setting, the linear
strategy from the previous section may be optimal. However, it does
not rule out the possibility for an increase in second moment after
one round to be compensated by a larger decrease later.

{\thm{Let $a^* = \sqrt{1 + \frac{1}{\sigma^2}}$ be as in
    Prop.~\ref{thm:secmomach}. Suppose that our multiplicative noise
    $Z_n$ has a Gaussian law $Z_n \sim \mathcal{N}(1,
    \sigma^2)$. Then, there exists $a > a^*$ for which a (non-linear)
    controller can stabilize the system in a second-moment
    sense. \label{thm:nonlinimprovement}} }

We first establish an elementary inequality for Gaussian variables. In
what follows, we define the signum function $\sgn(x)$ to be $1$ if $x
\ge 0$ and $-1$ otherwise.

\begin{lem} \label{lem:gaussian-sgn-bound}
  Let $Z \sim \mathcal{N}(1, \sigma^2)$, with $\sigma > 0$. We have
  \[ \bigE\left[ \text{\normalfont sgn}(Z) \left( 1 - \frac{Z}{1 + \sigma^2} \right) \right] > 0. \]
\end{lem}
\begin{proof}
  It is convenient to write $Z = 1 - \sigma \widetilde{Z}$, where $\widetilde{Z} \sim \mathcal{N}(0, 1)$. Let $s = \frac{1}{\sigma}$, and let $\gamma$ denote the standard Gaussian density. Note that
  \[ \gamma(x) \ge \frac{1}{\sqrt{2 \pi}} \max\left(1 - \frac{x^2}{2}, \;0 \right) \]
  for all $x$. Hence,
  \begin{align}
    \sqrt{2 \pi} \int_0^s \gamma(x) \,dx &\ge \int_0^s \max\left(1 - \frac{x^2}{2}, \;0 \right) \,dx \nonumber \\
    &= \begin{cases}
      s - \frac{s^3}{6} &\text{if $s \le \sqrt{2}$}\\
      \frac{2\sqrt{2}}{3} &\text{if $s > \sqrt{2}$}
    \end{cases} \label{eq:gamma-bound-1}
  \end{align}
  We also have
  \begin{equation} \label{eq:gamma-bound-2}
    s \sqrt{2 \pi} \int_s^\infty x \gamma(x) \,dx = s \int_s^\infty x e^{-\frac{x^2}{2}} \,dx = s e^{-\frac{s^2}{2}}.
  \end{equation}
  It can be checked by elementary calculations that for any $s > 0$, \eqref{eq:gamma-bound-1} is always strictly greater than \eqref{eq:gamma-bound-2}. Indeed, for $s < \sqrt{2}$ use 
  $e^{-\frac{s^{2}}{2}} < 1 - \frac{s^{2}}{2} + \frac{s^{4}}{8} < 1 - \frac{s^{2}}{6},$
  and for $s >\sqrt{2}$ note that $se^{-\frac{s^{2}}{2}}$ is decaying. Thus,
  \[ \frac{1}{2} - \bigE[ \mathbbm{1}_{\widetilde{Z} \ge s} ] = \int_0^s \gamma(x) \,dx > s \int_s^\infty x \gamma(x) \,dx = s \cdot \bigE \left[ \mathbbm{1}_{\widetilde{Z} \ge s} \cdot \widetilde{Z} \right] \]
  Let us rewrite the above equation in terms of $Z$ and $\sigma$, noting that \\ \mbox{$\mathbbm{1}_{\widetilde{Z} \ge s} = \frac{1}{2}(1 - \text{sgn}(Z))$}. We obtain
  \[ \bigE[\text{sgn}(Z)] > \frac{1}{\sigma} \cdot \bigE \left[ (1 - \text{sgn}(Z)) \cdot \frac{1 - Z}{\sigma} \right]. \]
  Rearranging, we have
  \[ \bigE\left[ \text{sgn}(Z) \left( 1 + \frac{1 - Z}{\sigma^2} \right) \right] > \frac{1}{\sigma^2} \bigE[1 - Z] = 0. \]
  Finally, multiplying both sides by $\frac{1 + \sigma^2}{\sigma^2}$ yields
  \[ \bigE\left[ \text{sgn}(Z) \left( 1 - \frac{Z}{1 + \sigma^2} \right) \right] > 0. \hfill\qedhere \] 
\end{proof}

\begin{proof}[Proof of Thm.~\ref{thm:nonlinimprovement}]
  To show second-moment stability, it suffices to exhibit controls $U_n$ and $U_{n +
    1}$ which ensure that $\bigE X_{n + 2}^2 \le \bigE X_n^2$ for all
  possible distributions of $X_n$. For a positive $\epsilon$ to be
  specified later, choose
  \[ a = (1 + \epsilon^2)a^* = (1 + \epsilon^2)\sqrt{1 + \frac{1}{\sigma^2}}. \]
  For our controls, we take
  \[ U_n = \frac{a}{1 + \sigma^{2}}Y_n \quad\text{and}\quad U_{n + 1} = \frac{a}{1 + \sigma^{2}} Y_{n + 1} + \epsilon Y_{n + 1} \cdot \left| \frac{Y_n}{Y_{n + 1}} \right|. \]
  Note that the expression for $U_n$ and the first term in the
  expression for $U_{n + 1}$ are the same as in the linear strategy
  from Prop.~\ref{thm:secmomach}. However, here we have added a small
  perturbation to $U_{n + 1}$. For convenience, define the function
  $g(x) = 1 - \frac{x}{1 + \sigma^2}$. Then,
  \begin{align}
    X_{n + 1} &= a \cdot g(Z_n) X_n. \nonumber \\
    X_{n + 2} &= a \cdot g(Z_{n + 1}) X_{n + 1} - \epsilon \cdot aZ_{n + 1}g(Z_n) \cdot \left| \frac{Z_n}{a Z_{n+1} g(Z_n)} \right| X_n \nonumber \\
    &= a^2 \cdot g(Z_{n+1})g(Z_n) X_n - \epsilon \cdot \text{sgn}(Z_{n+1}) \cdot g(Z_n) \left|\frac{Z_n}{g(Z_n)}\right| X_n. \label{eq:perturbedlinear}
  \end{align}
  We will compute the second moment of \eqref{eq:perturbedlinear}. Let
  \[ A = g(Z_{n+1})g(Z_n), \qquad B = \text{sgn}(Z_{n+1}) \cdot g(Z_n) \left|\frac{Z_n}{g(Z_n)}\right|. \]
  Then, we have
  \begin{align*}
    \bigE [A^2] &= \bigE[g(Z_{n+1})] \cdot \bigE[g(Z_n)] = \frac{\sigma^4}{(1 + \sigma^2)^2} \\
    \bigE [B^2] &= \bigE Z_n^2 = 1 + \sigma^2 \\
    \bigE [AB] &= \bigE \left[ g(Z_n)^2 \cdot \left| \frac{Z_n}{g(Z_n)} \right| \right] \cdot \bigE \Big[ \sgn(Z_{n + 1}) g(Z_{n + 1}) \Big] > 0,
  \end{align*}
  where the inequality in the last line follows from Lemma
  \ref{lem:gaussian-sgn-bound} and the fact that $g(Z_n)^2 \cdot
  \left| \frac{Z_n}{g(Z_n)} \right|$ is almost surely positive.

  Recall that the $Z_n$ and $Z_{n+1}$ are both independent of $X_n$, so taking second-moments in~\eqref{eq:perturbedlinear}, we have
  \begin{align*}
    \bigE X_{n + 2}^2 &= \left( a^4 \cdot \bigE A^2 - 2 \epsilon a^2 \cdot \bigE AB + \epsilon^2 \cdot \bigE B^2 \right) \bigE X_n^2 \nonumber \\
    &= \left( \frac{a^4 \sigma^4}{(1 + \sigma^2)^2} - 2 \epsilon a^2 \cdot \bigE AB + O(\epsilon^2) \right) \bigE X_n^2 \nonumber \\
   &= \Big( (1 + \epsilon^2)^4 - \frac{2 \epsilon (1 + \epsilon^2)(1 + \sigma^2)}{\sigma^2} \cdot \bigE AB + O(\epsilon^2) \Big) \bigE X_n^2 \nonumber \\
   & = \left[ 1 - \epsilon \cdot \frac{2 (1 + \sigma^2)}{\sigma^2} \bigE AB + O(\epsilon^2) \right] \bigE X_n^2.
  \end{align*}
  Since $\bigE AB > 0$, when $\epsilon$ is a sufficiently small
  positive number, this gives $\bigE X_{n + 2}^2 \le \bigE X_n^2$,
  showing second-moment stability.
\end{proof}

{\begin{remark} We actually suspect that
    Thm.~\ref{thm:nonlinimprovement} applies to all continuous distributions of
    $Z_n$. Indeed, the above analysis can be carried out for a more
    general class of control strategies. Consider instead
    \[ U_n = \frac{a}{1 + \sigma^{2}}Y_n \quad\text{and}\quad U_{n + 1} = \frac{a}{1 + \sigma^{2}} Y_{n + 1} + \epsilon Y_{n + 1} \cdot h\left(\frac{Y_n}{Y_{n + 1}}\right), \]
    where $h$ is any function (above, we used $h(x) = |x|$). Then, we
    would carry out the same analysis except with
    \[ B = aZ_{n + 1} g(Z_n) \cdot h \left( \frac{Z_n}{aZ_{n + 1}g(Z_n)} \right). \]
    The crucial properties we needed were that $\bigE B^2 < \infty$
    and $\bigE AB \ne 0$. Thus, for all distributions of $Z_n$, as long
    as there exists some function $h$ verifying those two properties,
    the conclusion of Thm.~\ref{thm:nonlinimprovement} applies.
\end{remark} \label{rem:nonlinimprovement}}

The next theorem shows that a perturbation can also improve upon
linear strategies when $\bigE Z_n = 0$.

{\thm{Suppose that instead of $\bigE Z_n = 1$, we have $\bigE Z_n =
    0$. Then, as long as $Z_n$ has finite second moment, there exists
    $a > 1$ for which a (non-linear) controller can stabilize the
    system in a second-moment sense. \label{thm:nonlinimprovement2}} }

We first prove a technical lemma.

\begin{lem} \label{lem:epsilon0}
  Let $Z$ be a random variable with $\bigE Z = 0$ and finite first
  moment. Then, for all sufficiently small $\epsilon > 0$, we have
  \[ \bigE \left[ Z \left| \frac{\epsilon}{Z} - 1 \right| \right] < 0. \]
\end{lem}
\begin{proof}
  For $0 \le t \le \frac{1}{2}$, define the function
  \[ f(x, t) = x \left| \frac{t}{x} - 1 \right|. \]
  Note that for each $x \ne 0$ and each $t$, we have
  \[ \left|\frac{f(x, t) - f(x, 0)}{t}\right| \le 1 \quad\text{and}\quad \lim_{t \rightarrow 0} \frac{f(x, t) - f(x, 0)}{t} = -1. \]
  Thus, letting $F(t) = \bigE f(Z, t)$, the dominated convergence theorem implies
  \[ \lim_{t \rightarrow 0} \frac{F(t) - F(0)}{t} = -1. \]
  Consequently, for all sufficiently small $t$, we have $F(t) < 0$, as
  desired.
\end{proof}

\begin{proof}[Proof of Thm.~\ref{thm:nonlinimprovement2}]
  We take an approach similar to the proof of
  Thm.~\ref{thm:nonlinimprovement}. Again, it suffices to exhibit
  controls $U_n$ and $U_{n + 1}$ which ensure that $\bigE \left[ X_{n + 2}^2\right]
  \le \bigE \left[X_n^2\right]$ for all possible distributions of $X_n$. By Lemma
  \ref{lem:epsilon0}, take a small enough $\epsilon_0 > 0$ so that
  \begin{equation} \label{eq:epsilon0-def}
    \bigE \left[ Z_n \left| \frac{\epsilon_0}{Z_n} - 1 \right| \right] < 0.
  \end{equation}
  Let $\epsilon > 0$ be another small number to be specified later,
  and take $a = 1 + \epsilon^2$. For our controls, we take
  \begin{align*}
    U_n &= a \epsilon_0^{-1} Y_n \\
    U_{n + 1} &= -a^2 \epsilon_0^{-1} Y_n - \epsilon Y_n \cdot \left| \frac{Y_{n + 1}}{Y_n} \right|.
  \end{align*}
  Then,
  \begin{align*}
    X_{n + 1} &= a X_n - a\epsilon_0^{-1} Z_n X_n. \nonumber \\
    X_{n + 2} &= a X_{n + 1} + a^2 \epsilon_0^{-1} Y_n + \epsilon Y_n \cdot \left| \frac{Y_{n + 1}}{Y_n} \right| \\
    &= a^2 X_n + \epsilon Z_n \cdot \left| \frac{aZ_{n + 1}(1 -  \epsilon_0^{-1} Z_n)}{Z_n} \right| X_n \\
    &= a^2 X_n + a \epsilon_0^{-1} \epsilon \cdot |Z_{n + 1}| \cdot Z_n \left| \frac{\epsilon_0}{Z_n} - 1 \right| X_n
  \end{align*}
  For convenience, let $A = \epsilon_0^{-1} \cdot |Z_{n + 1}| \cdot
  Z_n \left| \frac{\epsilon_0}{Z_n} - 1 \right|$, and note that $\bigE
  A^2 < \infty$ since $Z_n$ and $Z_{n + 1}$ have finite second
  moments. Substituting this definition for $A$, we calculate
  \begin{align*}
    \bigE \left[X_{n + 2}^2\right] &= a^2 \cdot \bigE (a + \epsilon A)^2 \cdot \bigE \left[X_n^2\right] \\
    &=(1 + \epsilon^{2})^2 \cdot \bigE (1 + \epsilon^{2} + \epsilon A)^2 \cdot \bigE \left[X_n^2\right] \\
    &= \left(1 + 2 \epsilon \cdot \bigE A + O(\epsilon^2)\right) \bigE \left[X_n^2\right].
  \end{align*}
  By \eqref{eq:epsilon0-def}, we have that $\bigE A$ is strictly
  negative. Thus, for small enough positive $\epsilon$, we obtain
  $\bigE \left[X_{n + 2}^2\right] \le \bigE \left[X_n^2\right]$, as desired.
\end{proof}

The next theorem pertains to schemes of the form $U_n = h(Y_n)$, where
$h: \mathbb{R} \to \mathbb{R}$ is \emph{any} fixed function.

{\thm{ Consider any $a > a^* = \sqrt{1 + \frac{1}{\sigma^2}}$ and any
    measurable function $h: \mathbb{R} \to \mathbb{R}$. Then, there exists a
    random variable $X$ with finite second moment for which
    \[ \bigE \left[a^2(X - h(XZ_n))^2\right] > \bigE X^2. \]
    In particular, we cannot guarantee $\bigE X_{n + 1}^2 \le \bigE
    X_n^2$ for the scheme $U_n = h(Y_n)$.
    \label{thm:nocontraction} }}
\begin{proof}
Let $M$ be a large parameter to be specified later. Consider the
  probability density
  \[ \rho(x) = \begin{cases}
     \left( 1 - \frac{1}{M^2} \right)^{-1} |x|^{-3} & \text{if $1 \le |x| \le M$}, \\
     0 & \text{otherwise}.
  \end{cases}. \]
  We will take $X$ to have density $\rho$, and for appropriate $M$, we will find that
  \[ \bigE \left[a^2(X - h(XZ_n))^2 \right] > \bigE X^2. \]

  Recall our notation $f_{Z}(x) = e^{-\phi(x)}$ for the density of $Z_n$. To aid in our calculations, for each integer $k \ge 0$ and real number $y \ne 0$, we consider the quantity
  \begin{align*}
    \alpha_k(y) &= \int_{-\infty}^\infty \frac{x^k \rho(x) f_{Z}(y/x)}{|x|} \,dx \\
    &= \int_1^M \frac{x^k \rho(x)f_{Z}(y/x) + (-x)^k \rho(-x)f_{Z}(-y/x)}{|x|} \,dx \\
    &= \left( 1 - \frac{1}{M^2} \right)^{-1} \int_1^M \frac{x^k f_{Z}(y/x) + (-x)^k f_{Z}(-y/x)}{x^4} \,dx \\
    &= \left( 1 - \frac{1}{M^2} \right)^{-1} \int_{y}^{y/M} \frac{y^ks^{-k} f_{Z}(s) + y^k(-s)^{-k} f_{Z}(-s)}{y^4s^{-4}} \left( - \frac{y}{s^2} \right) \,ds \\
    &= \left( 1 - \frac{1}{M^2} \right)^{-1} \int_{y/M}^{y} \frac{s^{2 - k} f_{Z}(s) + (-s)^{2 - k} f_{Z}(-s)}{y^{3 - k}} \,ds,
  \end{align*}
  where we have made the substitution $x = y/s$. Let $\epsilon > 0$ be a small parameter. Consider a fixed
  $t$ with $\epsilon \le t \le 1 - \epsilon$, and set $y = \pm
  M^t$. We find that
  \begin{equation} \label{eq:alpha-limit}
    \lim_{M \rightarrow \infty} |y| y^{2 - k} \alpha_k(y) = \int_0^\infty (s^{2 - k} f_{Z}(s) + (-s)^{2 - k} f_{Z}(-s)) \,ds = \bigE \left[Z_n^{2 - k}\right]
  \end{equation}
  uniformly over $\epsilon \le t \le 1 - \epsilon$, where we have
  taken care to ensure that the above holds for both possible signs of
  $y$. Let $\delta > 0$ also be a small parameter. We now choose $M$ to be sufficiently large so that
  \[ \left( 1 - \frac{1}{M^2} \right)^{-1} \le 1 + \delta, \]
  and also for all $y$ with $M^\epsilon \le |y| \le M^{1 - \epsilon}$ (in light of~\eqref{eq:alpha-limit}),
  \begin{align} \label{eq:alpha-bound}
    \alpha_2(y) - \frac{\alpha_1(y)^2}{\alpha_0(y)} &\ge (1 - \delta) \left( \frac{1}{|y|} - \frac{|y|^{-2}y^{-2} (\bigE Z_n)^2}{|y|^{-1}y^{-2} \bigE Z_n^2} \right) \\\
    &= (1 - \delta) \frac{1}{|y|} \left( 1 - \frac{1}{1 + \sigma^2} \right) \nonumber = (1 - \delta) \frac{\sigma^2}{|y|(1 + \sigma^2)}.
  \end{align}

  We then have
  \[ \bigE X^2 = \int_{-\infty}^\infty x^2 \rho(x) \,dx \le 2 (1 + \delta) \int_1^M \frac{1}{x} \,dx = 2 (1 + \delta) \log M \]
  and
  \begin{align*}
    \bigE (X - h(XZ_n))^2 &= \int_{-\infty}^\infty \int_{-\infty}^\infty \rho(x) (x - h(xz))^2 f_{Z}(z) \,dz\,dx \\
    &= \int_{-\infty}^\infty \int_{-\infty}^\infty \frac{\rho(x) (x - h(y))^2 f_{Z}(y/x)}{|x|} \,dy\,dx \\
    &= \int_{-\infty}^\infty \Big(\alpha_2(y) - 2h(y)\alpha_1(y) + h(y)^2\alpha_0(y) \Big) \,dy.
  \end{align*}
  Note that the integrand in the last expression is a quadratic function in $h(y)$ whose minimum possible value is $\alpha_2(y) - \frac{\alpha_1(y)^2}{\alpha_0(y)}$, and note also that this quantity is non-negative since $\alpha_2(y) \alpha_0(y) \ge \alpha_1(y)^2$ by the Cauchy-Schwarz inequality. Thus,
  \begin{align*}
    \bigE (X - h(XZ_n))^2 &\ge \int_{-\infty}^\infty \left(\alpha_2(y) - \frac{\alpha_1(y)^2}{\alpha_0(y)}\right) \,dy, \\
    &\ge \int_{M^{\epsilon}}^{M^{1 - \epsilon}} \left( \alpha_2(y) - \frac{\alpha_1(y)^2}{\alpha_0(y)} + \alpha_2(-y) - \frac{\alpha_1(-y)^2}{\alpha_0(-y)} \right) \,dy \\
    &\ge \frac{2 (1 - \delta) \sigma^2}{1 + \sigma^2} \int_{M^{\epsilon}}^{M^{1 - \epsilon}} \frac{1}{y} \,dy = \frac{2(1 - \delta)(1 - 2\epsilon) \sigma^2}{1 + \sigma^2} \log M,
  \end{align*}
  where we have plugged in the bound from~\eqref{eq:alpha-bound}. Consequently,
  \[ \frac{\bigE a^2(X - h(XZ_n))^2}{\bigE X^2} \ge \frac{(1 - \delta)(1 - 2\epsilon)}{1 + \delta} \cdot \frac{a^2\sigma^2}{1 + \sigma^2}. \]
  Since $a > \sqrt{1 + \frac{1}{\sigma^2}}$, the right hand side is
  strictly greater than $1$ when $\epsilon$ and $\delta$ are sufficiently
  small. This completes the proof.
\end{proof}

\section{An impossibility result}\label{sec:converse}

{\thm{For the system $\mathcal{S}_{a}$, suppose that $\phi$ is differentiable and satisfies $|z\cdot \phi'(z)| \leq C_{1} + C_{2} \cdot \phi(z)$ for all $z$, and also $e^{-\phi(z)} \leq |z|^{-1-\delta}$ for some $\delta > 0$. We additionally assume $\phi(\cdot)$ satisfies a doubling condition on $\phi'(\cdot)$ such that if $\frac{z_{1}}{2} \leq z_{2} \leq 2z_{1},$ then $\phi'(z_{2}) \leq C_{3} \cdot \phi'(z_{1})$.

Then, there exists $a \in \mathbb{R}$, $a < \infty$ such that $\mathbb{P}(|X_{n}|< M) \rightarrow 0$ for all $M < \infty$.

\label{thm:main}}}

Note that the conditions on $\phi(\cdot)$ above imply the conditions in Thm.~\ref{thm:intro1}.

We rewrite the system $\mathcal{S}_{a}$ from~\eqref{eq:systema} here, with state denoted as $X_{a,n}$, to emphasize the dependence on $a$:
\begin{align}
\begin{split}
X_{a, n+1} &= X_{a, n} - U_{a, n},  \\
Y_{a, n} &= Z_{n} \cdot X_{a, n}. \label{eq:systema1}
\end{split}
\end{align}
Now define $U_{n} := a^{-n} U_{a,n}$, and consider the system $\mathcal{S}$, which is the system $\mathcal{S}_{a}$ scaled by $a$:
\begin{align}
\begin{split}
X_{n+1} &= X_{n} - U_{n},  \\
Y_{n} &= Z_{n} \cdot X_{n}. \label{eq:systemnoa}
\end{split}
\end{align}
The $Z_{n}$'s and the initial state $X_{0} = X_{a,0}$ are identical in both systems.  Then, the scaled system satisfies $X_{n} = a^{-n} X_{a,n}$.  Thus we have that:
$$\mathbb{P}(|X_{a,n}| < M) = \mathbb{P}(|X_{n}|< a^{-n} M).$$

As a result it suffices bound the probability that the state of the of system $\mathcal{S}$, i.e. $|X_{n}|$, is contained in intervals that are shrinking by a factor of $a$ at each time step. The rest of this section uses the notation $X_{n}$ to refer to the state of the the system $\mathcal{S}$ and $X_{a.n}$ to refer the the state of the system $\mathcal{S}_{a}$.


\subsection{Definitions}

Let $S_{n} := \sum_{i=0}^{n-1}U_{i}$. Hence, $X_{n} = X_{0} - S_{n}$.

The goal of the controller is to have $S_{n}$ be as close to $X_{0}$ as possible. We will track the progress of the controller through intervals $I_{n}$ that contain $X_{0}$ and are decreasing in length.

Let $d(I_{n}, S) := \inf_{x\in I_{n}} |S-x|$ denote the distance of a point $S$ from the interval $I_{n}$.

\begin{defn}
For all $n \geq 0$ and for $k \in \mathbb{Z}$, there exists a unique integer $h(k)$ such that $X_{0} \in \bigr[\frac{h(k)}{2^{k}}, \frac{h(k)+1}{2^{k}}\bigr)$. Let $J(k) := \bigr[\frac{h(k)}{2^{k}}, \frac{h(k)+1}{2^{k}}\bigr)$. We now inductively define
    \[ K_{0} := \min\{k\geq0 \mid d(J(k), 0) \geq 2^{-k} \},~\textrm{and} \]
    \[ K_{n} := \min\{k \mid k > K_{n-1},~d(J(k), S_{n}) \geq 2^{-k} \}. \]
Write $H_{n} := h(K_{n})$ and $I_{n} := J(K_n) = \bigr[\frac{H_{n}}{2^{K_{n}}}, \frac{H_{n}+1}{2^{K_{n}}}\bigr)$.
\end{defn}
Let $Y_{0}^{n}$ indicate the observations $Y_{0}$ to $Y_{n}$, and let $\mathcal{F}_{n} := \{Y_{0}^{n}, K_{0}^{n}, H_{0}^{n} \}$, which is the total information available to the controller at time $n$. Let $f_{X_{n}}(x|\mathcal{F}_{n})$ be the conditional density of $X_{n}$ given $\mathcal{F}_{n}$.


\begin{figure}[th]
   \centering
   \includegraphics[width = .6\textwidth]{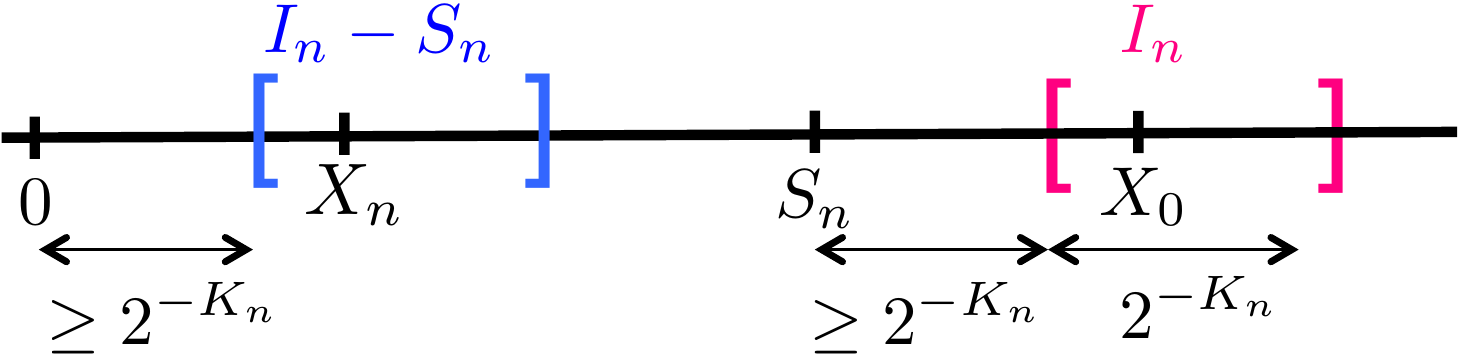} 
   \caption{A caricature illustrating the intervals $I_{n}$ and $I_{n}-S_{n}$.}
   \label{fig:intervals}
\end{figure}

\subsection{Relationships between $I_n$, $K_n$, $S_n$, and $X_n$}

We state and prove two lemmas that will be used in the main proof. The first lemma uses $K_{n}$ to bound how fast $S_{n}$ approaches $X_{0}$.
%
\begin{lem}
$$2^{-K_{n}} \leq |X_{0}-S_{n}|,$$
and if $K_{n} > K_{n-1}+1$, then
$$|X_{0}-S_{n}| \leq 2^{2-K_{n}}.$$
\label{lem:Inbound}
\end{lem}

\begin{proof}
From the definition of $I_{n}$, we know that $d(I_{n}, S_{n}) \geq 2^{-K_{n}} $. 
This gives $2^{-K_{n}} \leq |X_{0}-S_{n}|$, since $X_{0} \in I_{n}$.

To show the second half of the inequality, suppose that $|X_{0}-S_{n}| > 2^{2-K_{n}}$. Then,
\begin{align*}
2^{1-K_{n}}< |X_{0}-S_{n}| -2^{1-K_{n}}\,.
\end{align*}
Hence, there exists a larger interval $J(K_{n}-1)$ that contains $X_{0}$ such that
\begin{align*}
2^{1-K_{n}} < d(J(K_{n}-1), S_{n})\,,
\end{align*}
where $J(K_{n}-1)$ is an interval of length $2^{1-K_{n}}> 2^{-K_{n}}$. Since we also assumed that $K_{n} > K_{n-1}+1$, this contradicts the assumption that $K_{n}$ was the minimal $k>K_{n-1}$ such that $d(J(k), S_{n}) \geq 2^{-k}$.
\end{proof}

The second lemma bounds the ratio between two points in the interval of interest.

\begin{lem}
For $t \in I_{n} - S_{n}$ we have that $\frac{1}{2}\leq \frac{X_{n}}{t}\leq 2$.
\label{lem:ratiot}
\end{lem}

\begin{proof}
We have from Lemma~\ref{lem:Inbound} that $2^{-K_{n}} \leq |X_{n}|$. The lemma follows since the length of the interval $I_{n}-S_{n}$ is $2^{-K_{n}}$.
\end{proof}

\subsection{Preliminary estimates of the $Z_i$}

We also require some basic estimates for the $Z_i$, which we record here. Recall that we assumed the existence of a number $\delta > 0$ such that $e^{-\phi(z)} \le |z|^{-1-\delta}$.

\begin{lem} \label{lem:phi-tail}
  Let $\delta' = \delta/(1 + \delta)$. For each $i$ and any $t \ge 0$, we have
  \[ \mathbb{P}(\phi(Z_i) \ge t) \le \frac{2}{\delta'} e^{-\delta't}. \]
\end{lem}

\begin{proof}
 Let $s = e^{t/(1 + \delta)}$, so that $s^{-1 - \delta} = e^{-t}$. 
  We have
  \begin{align*}
    \mathbb{P}(\phi(Z_i) \ge t) &= 
    \mathbb{P}(e^{-\phi(Z_{i})} \le e^{-t}) \le \int_{-\infty}^\infty \min(e^{-t}, |z|^{-1 - \delta}) \,dz \\
    &= 2\left( \int_0^s e^{-t} \,dz + \int_s^\infty z^{-1 - \delta} \,dz \right) \\
    &= 2\left(s \cdot e^{-t} + \frac{s^{-\delta}}{\delta}\right) = 2\left( 1 + \frac{1}{\delta} \right) s^{-\delta} 
    = \frac{2}{\delta'} e^{-\delta' t}.
  \end{align*}\qedhere
\end{proof}

\begin{lem} \label{lem:phi-moments}
  For each $i$, the random variable $\phi(Z_i)$ has finite moments of all orders.
\end{lem}
\begin{proof}
  The condition $|Z_i \phi'(Z_i)| \le C_1 + C_2 \phi(Z_i)$ implies \mbox{$\phi(Z_i) \ge -\frac{C_1}{C_2}$.} According to Lemma~\ref{lem:phi-tail}, we also know that $\phi(Z_i)$ has exponentially decaying upper tails. Thus, $\phi(Z_i)$ has finite moments of all orders.
\end{proof}

\subsection{Proof of the main result}

The key element of the proof is to provide the interval $I_{n}$ to the controller at time $n$ as side-information in addition to $Y_{n}$. Our strategy is to first bound the density ${f}_{X_{n}}(x|\mathcal{F}_{n})$ by comparing the change in density from time $n$ to $n+1$. This bound helps us generate bounds for the probabilities of three events that cover the event of interest $\{|X_{n}|< a^{-n}M\}$. We will show that for large enough $a$ the probabilities of all three of these events go to $0$ as $n\rightarrow \infty$.

\begin{proof}[Proof of Thm.~\ref{thm:main}]

Consider
\begin{align*}
&~{f}_{X_{n}}(x \mid \mathcal{F}_{n}) \\
=&~f_{X_{n}}(x\mid Y_{n}, K_{n}, H_{n}, \mathcal{F}_{n-1}) \\
=&~\frac{f_{Y_{n}, K_{n}, H_{n}} (Y_{n}, K_{n}, H_{n}\mid X_{n} = x, \mathcal{F}_{n-1}) \cdot f_{X_{n}}(x\mid\mathcal{F}_{n-1})} {f_{Y_{n}, K_{n}, H_{n}}(Y_{n}, K_{n}, H_{n}\mid\mathcal{F}_{n-1})}
\end{align*}
Since $X_{0} \in I_{n}$, the controller knows that $X_{n} \in I_{n}-S_{n}$, where $I_{n}-S_{n}$ represents the interval $I_{n}$ shifted by $S_{n}$. We can calculate the ratio of the densities at  $x, w \in I_{n}-S_{n}$ as:
\begin{align}
&\frac{f_{X_{n}}(x\mid\mathcal{F}_{n})}{f_{X_{n}}(w\mid\mathcal{F}_{n})} \nonumber \\
=&\frac{f_{X_{n}}(x\mid\mathcal{F}_{n-1})}{f_{X_{n}}(w\mid\mathcal{F}_{n-1})} \cdot \frac{f_{Y_{n}} (Y_{n}\mid X_{n} = x, \mathcal{F}_{n-1})} {f_{Y_{n}} (Y_{n} \mid X_{n} = w, \mathcal{F}_{n-1})}.\label{eq:cancelY}
\end{align}

\noindent Since $K_{n}$ and $H_{n}$ are defined by $I_{n}$, the conditional distributions of $K_{n}$ and $H_{n}$ given $X_{n}=x$ and $X_{n}=w$ are equal for $x,w \in I_{n} - S_{n}$. So these terms cancel when we consider a ratio, giving~\eqref{eq:cancelY}.

Taking logarithms and using the triangle inequality gives the following recursive lemma. 
\begin{lem} \label{lem:ratio}
\begin{align}
\biggr| \log \frac{f_{X_{n}}(x \mid \mathcal{F}_{n})}{f_{X_{n}}(w \mid \mathcal{F}_{n})}\biggr|~\leq 2^{K_{n}+1}C_{3}\left|Z_{n}\cdot \phi'\left(Z_{n}\right)\right||x-w| + \biggr|\log \frac{f_{X_{n}}(x\mid\mathcal{F}_{n-1})}{f_{X_{n}}(w\mid\mathcal{F}_{n-1})}\biggr|. \label{eq:beforepsi}
\end{align}
\end{lem}

The proof is deferred to Section~\ref{sec:appendix2} to improve readability.
This lemma helps us establish the recursive step, since the control law gives us that:
\begin{align*}
f_{X_{n}}(x \mid \mathcal{F}_{n-1}) = f_{X_{n-1}}(x + U_{n-1}\mid \mathcal{F}_{n-1}),
\end{align*}
since $U_{n-1}$ is $\mathcal{F}_{n-1}$ measurable. Substituting this into~\eqref{eq:beforepsi} and unfolding recursively gives:
\begin{align}
&\biggr| \log \frac{f_{X_{n}}(x\mid\mathcal{F}_{n})}{f_{X_{n}}(w\mid\mathcal{F}_{n})}\biggr| \leq \label{eq:recursion}\\
&\sum_{i=1}^{n} 2^{K_{i}+1}C_{3}\left|Z_{i}\cdot \phi'\left(Z_{i}\right)\right||x-w| + \biggr|\log \frac{f_{X_{0}}(x + S_{n})}{f_{X_{0}}(w  + S_{n})}\biggr|. \nonumber
\end{align}
The inequality~\eqref{eq:recursion} separates the effect of the uncertainty due to $X_{0}$ and the subsequent uncertainty due to the observations and control.

Let $\eta_n = \max_{x, w\in I_n - S_n} \left|\log \frac{f_{X_{0}}(x+S_{n})}{f_{X_{0}}(w+S_{n})}\right|$. Since $I_n$ is an interval of size at most $2^{-n}$ which contains $X_0$, we get that
\begin{equation}\label{eta-n-bound}
|\eta_n| \leq \frac{1}{2} |(X_0 +  2^{2-n})^2 - (X_0 - 2^{2-n})^2| \leq  2^{3-n} |X_0|\,.
\end{equation}
Now, we define
\begin{equation}\label{eq-Psi-n}
\Psi_{n} = \sum_{i=0}^{n} 2^{K_{i}+1}C_{3}\left|Z_{i}\cdot \phi'\left(Z_{i}\right)\right| 2^{-K_{n}},
\end{equation}
and rewrite~\eqref{eq:recursion} as:
\begin{align*}
\biggr| \log \frac{f_{X_{n}}(x\mid\mathcal{F}_{n})}{f_{X_{n}}(w\mid\mathcal{F}_{n})}\biggr|
\leq \Psi_n \cdot 2^{K_{n}} \cdot |x-w| + \eta_n.
\end{align*}

We will need the following lemma to bound the crucial quantity $\Psi_{n}$.
\begin{lem} \label{lem:psilem}
For a sufficiently large constant $T$, the expectation $\bigE[e^{\Psi_n 2^{-T}}]$ is uniformly bounded for all $n$.
\end{lem}
The proof of this lemma is deferred to Section~\ref{sec:appendix2}. Henceforth, let $T$ denote a constant that is sufficiently large for Lemma~\ref{lem:psilem} to apply.

Finally, we are in a position to get a bound on $f_{X_n}(x\mid \mathcal{F}_{n})$:
\begin{align}
f_{X_n}(x\mid \mathcal{F}_{n}) \leq~(e^{\Psi_n 2^{K_{n}} |x-w| +\eta_n}) f_{X_{n}}(w\mid\mathcal{F}_{n}).\label{eq:prelimbound2}
\end{align}
Now, we integrate~\eqref{eq:prelimbound2} over an interval of length $\gamma = 2^{(-K_{n}-T)}$ with $x$ at one end point. So $|x-w| \leq 2^{(-K_{n}-T)}$. Such an interval can be fit into $I_{n}$ to the left or right of any $x$ depending on where $x$ is in the interval. Assuming without loss of generality that $x$ is the left endpoint of the integration interval we compute that
\begin{align*}
\int_{x}^{x+\gamma} f_{X_n}(x| \mathcal{F}_n)~dw \leq&\int_{x}^{x+\gamma}  (e^{\Psi_n 2^{K_{n}} |x-w| +\eta_n}) f_{X_n}(w|\mathcal{F}_n) ~dw.
\end{align*}
We bound $|x-w|$ on the RHS by $\gamma = 2^{(-K_{n}-T)}$ to get
\begin{align*}
\gamma \cdot f_{X_n}(x| \mathcal{F}_n)
 \leq&\int_{x}^{x+\gamma}  (e^{\Psi_n 2^{K_{n}} 2^{(-K_{n}-T)} +\eta_n}) f_{X_n}(w|\mathcal{F}_n) ~dw\\
\leq&~(e^{\Psi_n 2^{-T}+\eta_n}) \cdot 1.
\end{align*}
The last step follows since the density integrates out to $1$. Hence,
\begin{align}
f_{X_n}(x| \mathcal{F}_n) \leq&~e^{\Psi_n 2^{-T}+\eta_n}~2^{K_{n}+T}.\label{eq:fnbound}
\end{align}
This gives us a bound on the density of $X_{n}$ in terms of $K_{n}$.

It now remains to bound the rate at which the $K_{n}$ are growing. The following lemma shows that the $K_{n}$ grow essentially at most linearly.
{\lem{
There exists a constant $C$ such that
\begin{align*}
\mathbb{P}(K_{n} - K_{0}> C \cdot n) \rightarrow 0 \text{ as } n \rightarrow \infty.
\end{align*}
\label{lem:mn}}}
\begin{proof}
By construction, $K_{n+1} \geq K_{n}+1$. In the case where $K_{n+1}>K_{n} +1$, we can apply Lemma~\ref{lem:Inbound} and get that for $\ell \geq 2$
\begin{align*}
\mathbb{P}(K_{n+1}-K_{n} \geq \ell|\mathcal{F}_{n})\leq&~\mathbb{P}(|X_{n+1}|\leq 2^{2-K_{n}-\ell} |\mathcal{F}_{n}) \\
=&~\mathbb{P}(|X_{n} - U_{n}|\leq 2^{2-K_{n}-\ell} |\mathcal{F}_{n})\,.
\end{align*}
This is because the control $U_{n}$ must have been very close to $X_{n}$ for $K_{n+1}$ to be much larger than $K_{n}$. Then we calculate this probability by integrating out the density as:
\begin{align*}
\mathbb{P}(|X_{n} - U_{n}|\leq 2^{2-K_{n}-\ell} |\mathcal{F}_{n})
=&\int_{U_{n}-2^{2-K_{n}-\ell}}^{U_{n}+2^{2-K_{n}-\ell}} f_{X_{n}}(t|\mathcal{F}_{n})~dt\\
\leq&~2 \cdot 2^{2-K_{n}-\ell} \left( \max_{t}f_{X_{n}}(t|\mathcal{F}_{n})\right).
\end{align*}
Combined with \eqref{eq:fnbound}, this gives us that
\begin{align}
\mathbb{P}(K_{n+1}-K_{n} \geq \ell|\mathcal{F}_{n})
\leq~&2 \cdot 2^{2-K_{n}-\ell} ~e^{\eta_n+\Psi_n 2^{-T}}~2^{K_n+T}\nonumber\\
=~&2^{3-\ell+T} \cdot e^{\eta_n+\Psi_n 2^{-T}}. \label{eq:knbound}
\end{align}
Write $D_n = K_{n+1} - K_n$, and let
$$\widetilde{K}_n = \sum_{i=0}^{n-1} (D_i - \mathbb E[D_i \mid \mathcal F_i]) \,.$$
It is clear that $(\widetilde{K}_n)$ is a martingale with respect to $\mathcal F_n$. In addition, \eqref{eq:knbound}  yields that the conditional distribution of $D_n$ given $\mathcal F_n$ is stochastically dominated by the distribution of
\begin{equation}\label{domination}
G_n + \left(3+T + \frac{|\eta_n + \Psi_n 2^{-T}|}{\log 2}\right)\,,
\end{equation}
where $G_n$ is an independent geometric variable with mean 2.

By \eqref{eta-n-bound} and Lemma~\ref{lem:psilem}, both $\eta_n$ and $\Psi_n 2^{-T}$ have bounded second moments, and so for some constant $\widetilde{C}$, we have
\[\mathbb E (D_i - \mathbb E[D_i \mid \mathcal F_i])^2 \leq \mathbb E D_i^2 \leq \widetilde{C}\,.\]
Summing over $i$, this implies that $\bigE[\widetilde K_n^2] \le \widetilde{C}n$, and so
\begin{equation}\label{eq:tilde-K}
\mathbb{P}(\widetilde K_n \geq n) \to 0 \mbox{ as } n\to \infty\,.
\end{equation}

We now turn our attention to terms of the form $\bigE[D_i \mid \mathcal{F}_i]$. Using \eqref{domination} again, we get that
\begin{equation}\label{eq-D-i}
\mathbb E[D_i \mid \mathcal F_i] \leq 5 + T + 2|\eta_i + \Psi_i 2^{-T}|.
\end{equation}
Observe that from the definition of $\Psi_i$ given in \eqref{eq-Psi-n}, we have
\begin{align} \label{eq-Delta-sum}
  \sum_{i = 0}^n \Psi_i &\le C_3 \sum_{i = 0}^n \sum_{j = 0}^i 2^{K_j - K_i + 1} |Z_j \cdot \phi'(Z_j)| \nonumber \\
  &= C_3 \sum_{j = 0}^n \sum_{i = j}^n 2^{K_j - K_i + 1} |Z_j \cdot \phi'(Z_j)| \nonumber \\
  &\le 4 C_3 \sum_{j = 0}^n |Z_j \cdot \phi'(Z_j)|,
\end{align}
where in the last step we have used the fact that the $K_i$ increase by at least $1$ in each step, so that $K_i - K_j \ge i - j$. Then, applying the bound $|Z_j \cdot \phi'(Z_j)| \le C_1 + C_2 \phi(Z_j)$ to \eqref{eq-Delta-sum} yields
\[ \sum_{i = 0}^n \Psi_i \le 4C_3C_1n + 4C_3C_2 \sum_{i = 1}^n \phi(Z_i). \]
Summing \eqref{eq-D-i} over $i$ and applying the above bound gives
\begin{align*}
  \sum_{i = 1}^n \bigE[D_i \mid \mathcal F_i] &\le (5 + T)n + 2 \sum_{i = 1}^n |\eta_i| + 2^{1 - T} \sum_{i = 1}^n \Psi_i \\
  &\le C_{D,1} \left( n + \sum_{i = 1}^n |\eta_i| + \sum_{i = 1}^n \phi(Z_i) \right)
\end{align*}
for a constant $C_{D,1}$. Now, recalling \eqref{eta-n-bound}, we see that the quantity
\[ \sum_{i = 1}^n |\eta_i| \le 8|X_0| \]
has mean and variance bounded by a constant, which we call $C_\eta$. In addition, by Lemma~\ref{lem:phi-moments}, there exists another constant $C_\phi$ which upper bounds the mean and variance of $\phi(Z_i)$. We conclude that
\begin{align*}
  \bigE\left( \sum_{i = 1}^n \bigE[D_i \mid \mathcal F_i] \right) &\le C_{D,1}(1 + C_\eta + C_\phi \cdot n)  \\
  \mathrm{Var}\left( \sum_{i = 1}^n \bigE[D_i \mid \mathcal F_i] \right) &\le 2C_{D,1}^2(C_\eta + C_\phi \cdot n)
\end{align*}
It follows that
\begin{equation} \label{eq:E[D_i|F_i]}
\mathbb{P}\left(\sum_{i = 1}^n \bigE[D_i \mid \mathcal F_i] \ge C_{D,2} \cdot n \right) \to 0 \mbox{ as } n\to \infty\,,
\end{equation}
where $C_{D,2} = C_{D,1} C_\phi + 1$. Finally, setting $C = C_{D,2} + 1$, we have
\begin{align*}
  \mathbb{P}\left( K_n - K_0 > C \cdot n \right) &= \mathbb{P}\left( \sum_{i = 1}^n D_i > C \cdot n \right) \\
  &= \mathbb{P}\left( \widetilde{K}_n + \sum_{i = 1}^n \bigE[D_i \mid \mathcal{F}_i] > (C_{D, 2} + 1) \cdot n \right) \\
  &\le \mathbb{P}(\widetilde{K}_n > n) + \mathbb{P}\left(  \sum_{i = 1}^n \bigE[D_i \mid \mathcal F_i] \ge C_{D,2} \cdot n \right),
\end{align*}
where the last expression goes to $0$ as $n \rightarrow \infty$ by \eqref{eq:tilde-K} and \eqref{eq:E[D_i|F_i]}.
\end{proof}

This bound on the growth of the $K_{n}$ variables allows us to complete the proof of Thm.~\ref{thm:main}.

Let $G_{n}$ denote the event that $K_{n}-K_{0}> Cn$, and $G_{n}^{c}$ its complement. Then we can cover the event of interest by three events, and get that
\begin{align}
&\mathbb{P}(|X_{n}| < a^{-n}M)  \label{eq:events}\\
\leq&~\mathbb{P}(G_{n}) + \mathbb{P}(K_{0} >n) + \mathbb{P}(|X_{n}| \leq a^{-n}M, G_{n}^{c}, K_{0}\leq n). \nonumber
\end{align}
We evaluate the three terms one by one. For the first term in~\eqref{eq:events}, we have $\mathbb{P}(G_{n}) = \mathbb{P}(K_{n}- K_{0} >Cn) \rightarrow 0$ as $n\rightarrow \infty$ from Lemma~\ref{lem:mn}.

The second term, $\mathbb{P}(K_{0} > n)$, captures the case where the initial state $X_{0}$ might be very close to zero. However, eventually this advantage dies out for large enough $n$, since $\mathbb{P}(X_{0} < 2^{-n}) \rightarrow 0$ as $n\rightarrow \infty$.

The last term in~\eqref{eq:events} remains. By the law of iterated expectation:
\begin{align*}
\mathbb{P}(|X_{n}| < a^{-n}M,~G_{n}^{c}, K_{0} \leq n)
=\bigE[\mathbb{P}(|X_{n}| < a^{-n}M,~G_{n}^{c}, K_{0} \leq n \mid \mathcal{F}_{n})].
\end{align*}
We focus on the term conditioned on $\mathcal{F}_{n}$:
\begin{align}
\mathbb{P}(|X_{n}| < a^{-n}M, G_{n}^{c}, K_{0} \leq n \mid \mathcal{F}_{n})
&=\bigE[\mathbbm{1}_{\{|X_{n}| < a^{-n}M\}} \mathbbm{1}_{\{G_{n}^{c}\}} \mathbbm{1}_{\{K_{0}\leq n\}} \mid\mathcal{F}_{n}] \nonumber \\
&=\mathbb{P}(|X_{n}| < a^{-n}M \mid \mathcal{F}_{n})\cdot \mathbbm{1}_{\{G_{n}^{c}\}} \mathbbm{1}_{\{K_{0}\leq n\}}. \label{eq:firstbound}
\end{align}

Now, we can apply~\eqref{eq:fnbound} to get
\begin{align*}
\mathbb{P}(|X_{n}| < a^{-n}M|\mathcal{F}_{n})
=& \int_{-a^{-n}M}^{a^{-n}M} f_{X_{n}}(x|\mathcal{F}_{n})~dx \nonumber \\
\leq& \int_{-a^{-n}M}^{a^{-n}M}e^{\eta+\Psi_n 2^{-T}}~2^{K_{n}+T}~dx \nonumber \\
=&~2Ma^{-n}\cdot e^{\eta_n+\Psi_{n}2^{-T}}\cdot 2^{K_{n}+T}.
\end{align*}

\noindent Then we can bound~\eqref{eq:firstbound} as
\begin{align*}
&\mathbb{P}(|X_{n}| < a^{-n}M|\mathcal{F}_{n}) \cdot \mathbbm{1}_{\{G_{n}^{c}\}} \mathbbm{1}_{\{K_{0}\leq n\}}
\leq&~2Ma^{-n}\cdot e^{\eta_n+\Psi_{n}2^{-T}}\cdot 2^{(C+1)n +T}.
\end{align*}
since $K_{n} \leq Cn + K_{0}$ and $K_{0}\leq n$ implies $K_{n} \leq (C+1)n$. Taking expectations on both sides we get:
\begin{align}
&\mathbb{P}(|X_{n}| < a^{-n}M,~G_{n}^{c}, K_{0} \leq n) \leq&~2Ma^{-n}\cdot 2^{(C+1)n +T}\cdot\bigE[e^{\eta_n} e^{\Psi_{n}2^{-T}}]\,.\label{eq:term3}
\end{align}
By Lemma~\ref{lem:psilem} and \eqref{eta-n-bound}, the above expression~\eqref{eq:term3} tends to 0  for $a > 2^{C+1}$.

Thus, all three probabilities in \eqref{eq:events} converge to $0$ as $n \rightarrow \infty$. Hence, if $a > 2^{C+1}$ then $\mathbb{P}(|X_{n}| < a^{-n}M) \rightarrow 0$ for all $M$.
\end{proof}


\section{Bounding the likelihood ratio} \label{sec:appendix2}
Here we provide the proofs of two lemmas used to bound the term $\biggr| \log \frac{f_{X_{n}}(x\mid\mathcal{F}_{n})}{f_{X_{n}}(w\mid\mathcal{F}_{n})}\biggr|$.

\subsection{Proof of Lemma~\ref{lem:ratio}}

We take logarithms on both sides of~\eqref{eq:cancelY} and apply the triangle inequality to get
\begin{align}
&\biggr| \log \frac{f_{X_{n}}(x \mid\mathcal{F}_{n})}{f_{X_{n}}(w \mid \mathcal{F}_{n})}\biggr| \nonumber \\
\leq& \biggr|\log \frac{f_{Y_{n}} (Y_{n} \mid X_{n} = x, \mathcal{F}_{n-1})} {f_{Y_{n}} (Y_{n} \mid X_{n} = w, \mathcal{F}_{n-1})} \biggr| + \biggr|\log \frac{f_{X_{n}}(x \mid \mathcal{F}_{n-1})}{f_{X_{n}}(w\mid\mathcal{F}_{n-1})}\biggr|. \label{eq:splitY}
\end{align}
The form of the density of $Z$ gives,
\begin{align}
\biggr|\log \frac{f_{Y_{n}} (Y_{n}| X_{n} = x, \mathcal{F}_{n-1})} {f_{Y_{n}} (Y_{n} | X_{n} = w, \mathcal{F}_{n-1})} \biggr| = \left| \phi\left(\frac{Y_{n}}{x}\right) - \phi\left(\frac{Y_{n}}{w}\right)\right|.\label{eq:derivbound}
\end{align}

We can use the derivatives of the functions to bound the two function differences above. Since $\frac{d}{dx} \phi\left(\frac{Y_{n}}{x}\right) = \frac{Y_{n}}{x^{2}}\phi'\left(\frac{Y_{n}}{x}\right)$, we bound~\eqref{eq:derivbound} as below. Since $X_{n} \in I_{n}-S_{n}$, the maximizations are over $t\in I_{n}-S_{n}$.
\begin{align}
\left| \phi\left(\frac{Y_{n}}{x}\right) - \phi\left(\frac{Y_{n}}{w}\right)\right|
\leq& \max_{t \in I_{n}-S_{n}} \left|\frac{Y_{n}}{t^{2}}\phi'\left(\frac{Y_{n}}{t}\right)\right||x-w|.\label{eq:firstmaxsplit}
\end{align}
For all $t \in I_{n}-S_{n}$, by Lemma~\ref{lem:ratiot}, we have $\frac{1}{2}\leq \frac{X_{n}}{t}\leq 2$. Using this and $Y_{n} = Z_{n}X_{n}$, we get the following bound on~\eqref{eq:firstmaxsplit}:
\begin{align}
\left| \phi\left(\frac{Y_{n}}{x}\right) - \phi\left(\frac{Y_{n}}{w}\right)\right|
\leq&\max_{t \in I_{n}-S_{n}} \left|\frac{2Z_{n}}{t}\phi'\left(\frac{Z_{n}X_{n}}{t}\right)\right||x-w|\nonumber\\
\leq&\max_{t \in I_{n}-S_{n}} 2C_{3}\left|\frac{Z_{n}}{t}\phi'\left(Z_{n}\right)\right||x-w|.\label{eq:doubling}
\end{align}
\eqref{eq:doubling} follows from the doubling property of $\phi'(\cdot)$, since $\frac{Z_{n}X_{n}}{t}$ and $Z_{n}$ are within a factor of two from each other by Lemma~\ref{lem:ratiot}. Now note that
$$\max_{t\in I_{n}-S_{n}}\frac{1}{|t|} \leq 2^{K_{n}}.$$
Applying this to the bound from~\eqref{eq:doubling} we get:
\begin{align*}
&\left| \phi\left(\frac{Y_{n}}{x}\right) - \phi\left(\frac{Y_{n}}{w}\right)\right| \leq 2^{K_{n}+1}C_{3}\left|Z_{n}\cdot \phi'\left(Z_{n}\right)\right||x-w|.
\end{align*}
This now gives a bound for~\eqref{eq:splitY} as below:
\begin{align}
\biggr| \log \frac{f_{X_{n}}(x \mid \mathcal{F}_{n})}{f_{X_{n}}(w \mid \mathcal{F}_{n})}\biggr| \leq
2^{K_{n}+1}C_{3}\left|Z_{n}\cdot \phi'\left(Z_{n}\right)\right||x-w| + \biggr|\log \frac{f_{X_{n}}(x\mid\mathcal{F}_{n-1})}{f_{X_{n}}(w\mid\mathcal{F}_{n-1})}\biggr|.
\end{align}
\qed

\subsection{Proof of Lemma~\ref{lem:psilem}}

Recall that our goal is to estimate the quantity
\[ \Psi_{n} = \sum_{i=0}^{n} 2^{K_{i}+1}C_{3}\left|Z_{i}\cdot \phi'\left(Z_{i}\right)\right| 2^{-K_{n}}. \]
Since the $K_i$'s must increase by at least one in each step, we have
$K_n - K_i \ge n - i$, and so
\begin{align}
  \Psi_n &\le \sum_{i=0}^{n} 2^{1 + i - n}C_{3}\left|Z_{i}\cdot \phi'\left(Z_{i}\right)\right|
  \le \sum_{i=0}^{n} 2^{1 + i - n}C_{3}\left(C_{1} + C_{2} \phi\left(Z_{i}\right)\right) \nonumber \\
  &= 4C_3C_1 + C_3C_2 \sum_{i=0}^n 2^{1 + i - n} \phi(Z_{i}), \label{eq:Psi-bound}
\end{align}
where we have also used the assumption $|z\cdot \phi'(z)| \leq C_{1} + C_{2} \phi(z)$.

Let $\delta' = \delta/(1 + \delta)$ as in Lemma~\ref{lem:phi-tail}. Consider any $\theta < \delta'/2$. Applying Lemma~\ref{lem:phi-tail}, we have for each $i$ that
\begin{align}
  \bigE[e^{\theta \phi(Z_i)}] &= \int_{-\infty}^\infty \theta e^{\theta t} \cdot \mathbb{P}\left( \phi(Z_i) \ge t \right) \,dt \nonumber \\
  &\le \int_{-\infty}^0 \theta e^{\theta t} \,dt + \frac{2 \theta}{\delta'} \int_0^\infty e^{(\theta - \delta')t}\,dt \nonumber \\
  &= 1 + \frac{2 \theta}{\delta'(\delta' - \theta)} \le 1 + \frac{4 \theta}{\delta'^2}. \label{eq:theta-bound}
\end{align}

Now, choose $T$ large enough so that $2^{1 - T} C_2C_3 < \delta'/2$. We can then apply \eqref{eq:theta-bound} to each term in \eqref{eq:Psi-bound} by taking $\theta = 2^{-T}C_3C_2 \cdot 2^{1 + i - n}$. This yields
\[ \bigE\left[ \exp\left( 2^{-T}C_3C_2 \cdot 2^{1 + i - n} \phi(Z_i) \right) \right] \le 1 + C \cdot 2^{i - n} \]
for a constant $C$ not depending on $n$. We then have
\begin{align*}
  \bigE[e^{2^{-T} \Psi_n}] &\le \bigE\left[ \exp\left( 4C_3C_1 + C_3C_2 \sum_{i=0}^n 2^{1 + i - n} \phi(Z_{i}) \right) \right] \\
  &\le e^{4C_3C_1} \prod_{i=0}^n \left(1 + C \cdot 2^{i - n} \right) \\
  &\le e^{4C_3C_1} \prod_{i=0}^\infty \left(1 + C \cdot 2^{-i} \right),
\end{align*}
which is a (finite) constant not depending on $n$.
\qed

\section{Conclusion}
This paper provides a first proof-of-concept converse for a control system observed over continuous multiplicative noise. However, there is an exponential gap between the scaling behavior of the achievable strategy and the converse. 

We note that if the system $\mathcal{S}_{a}$ in~\eqref{eq:systema} is restricted to using linear control strategies, then its performance limit is the same as that of a system with the same multiplicative actuation noise (i.e. the control $U_{n}$ is multiplied by a random scaling factor) but perfect observations (as in~\cite{controlcapacity}). Previous work has shown how to compute the control capacity for 
systems with multiplicative noise on the actuation channel~\cite{controlcapacity, gireejabeast}. However, computing the control capacity of the system $\mathcal{S}_{a}$, i.e. computing tight upper and lower bounds on the system growth factor $a$, remains open.

\bibliographystyle{abbrv}
\bibliography{tigerjournal.bib}



\end{document}